\documentclass[12pt,a4paper]{article}

\usepackage{amsmath,amsfonts,amssymb,amsthm}
\usepackage{mathtools}
\usepackage[utf8]{inputenc}
\usepackage[osf,noBBpl]{mathpazo}
\usepackage{microtype}
\usepackage{xcolor}
\definecolor{darkgreen}{rgb}{0,0.6,0}

\usepackage[pdftex,bookmarks=true,bookmarksnumbered=true,pdfpagemode=UseOutlines,pdfstartview=FitH,%
    pdfpagelayout=OneColumn,colorlinks=true,urlcolor=blue,citecolor=darkgreen,pdfborder={0 0 0},ocgcolorlinks]{hyperref}

\usepackage{amsmath,amsfonts,amssymb,amsthm}
\usepackage{mathtools}
\usepackage{bm}

\theoremstyle{plain}
\newtheorem{theorem}{Theorem}
\newtheorem{lemma}[theorem]{Lemma}
\newtheorem{proposition}[theorem]{Proposition}
\newtheorem{corollary}[theorem]{Corollary}
\newtheorem{algo}{Algorithm}

\theoremstyle{definition}
\newtheorem{definition}[theorem]{Definition}

\newcommand\refsec[1]{Sec.~\ref{sec:#1}}

\newcommand\KK{\mathbb K}
\newcommand\QQ{\mathbb Q}
\newcommand\CC{\mathbb C}
\newcommand\ZZ{\mathbb Z}
\newcommand\NN{\mathbb N}

\renewcommand\P{\mathsf P}
\newcommand\NP{\mathsf{NP}}

\newcommand\closure\overline
\newcommand\mult[2]{\operatorname{mult}_{#1}(#2)}
\newcommand\puiseux[2][\closure\KK]{#1\langle\!\langle#2\rangle\!\rangle}

\newcommand\divides{|}

\newcommand\I{\mathcal I}
\newcommand\J{\mathcal J}
\renewcommand\O{\mathcal O}

\DeclareMathOperator\val{val}
\DeclareMathOperator\Newt{Newt}
\DeclareMathOperator\Supp{Supp}
\DeclareMathOperator\res{res}
\renewcommand\wr{\operatorname{wr}}
\DeclareMathOperator\size{size}

\title{Computing low-degree factors of lacunary polynomials:\\ a Newton-Puiseux approach}

\author{Bruno Grenet\thanks{Supported by the LIX-Qualcomm-Carnot fellowship.}\\
    LIX -- UMR 7161\\
    \'Ecole Polytechnique\\
    91\,128 Palaiseau Cedex, France\\
    \href{mailto:bruno.grenet@lix.polytechnique.fr}{\nolinkurl{bruno.grenet@lix.polytechnique.fr}}
}
\date{May 8, 2014}
\makeatletter
\hypersetup{
    pdftitle = {\@title},
    pdfauthor = {Bruno Grenet}
}
\makeatother

\begin{document}

\maketitle

\begin{abstract}
We present a new algorithm for the computation of the irreducible factors of degree at most $d$, with multiplicity, of multivariate lacunary polynomials over fields of characteristic zero. The algorithm reduces this computation to the computation of irreducible factors of degree at most $d$ of univariate lacunary polynomials and to the factorization of low-degree multivariate polynomials. The reduction runs in time polynomial in the size of the input polynomial and in $d$. As a result, we obtain a new polynomial-time algorithm for the computation of low-degree factors, with multiplicity, of multivariate lacunary polynomials over number fields, but our method also gives partial results for other fields, such as the fields of $p$-adic numbers or for absolute or approximate factorization for instance.

The core of our reduction uses the Newton polygon of the input polynomial, and its validity is based on the Newton-Puiseux expansion of roots of bivariate polynomials. In particular, we bound the valuation of $f(X,\phi)$ where $f$ is a lacunary polynomial and $\phi$ a Puiseux series whose vanishing polynomial has low degree. 
\end{abstract}

\section{Introduction} 

This article proposes a new algorithm for computing low-degree factors of lacunary polynomials over fields of characteristic $0$. The \emph{lacunary representation} of a polynomial 
\[f(X_1,\dotsc,X_n)=\sum_{j=1}^k c_j X_1^{\alpha_{1,j}}\dotsb X_n^{\alpha_{n,j}}\]
is the list $\{(c_j,\alpha_{1,j},\dotsc,\alpha_{n,j}):1\le j\le k\}$. We define the lacunary size of $f$, denoted by $\size(f)$, as the size of the binary representation of this list. It takes into account the size of the coefficients, and thus depends on the field they belong to. An important remark is that the size is proportional to the logarithm of the degree. 

Over algebraic number fields, the factorization problem can be solved in time polynomial in the degree of the input polynomial (see for instance~\cite{vzGaGe03} and references therein). It is also the case of absolute factorization, that is factorization over the algebraic closure of $\QQ$~\cite{CheGa05}. In the case of lacunary polynomials, these algorithms are not adapted since they are exponential in the size of the representation. 

Actually, the computation of the irreducible factorization of a polynomial given in lacunary representation cannot be performed in polynomial time. For instance over $\QQ$, the polynomial $X^p-1$ has a size of order $\log(p)$, while one of its irreducible factors, namely $(1+X+\dotsb+X^{p-1})$, has a size of order $p$. 

Therefore, a natural restriction consists in computing low-degree factors only. A line of work yielded an algorithm that, given a lacunary polynomial $f\in\KK[X_1,\dotsc,X_n]$ and an integer $d$ as input, where $\KK$ is an algebraic number field, computes all the irreducible factors of $f$ of degree at most $d$ in time polynomial in $\size(f)$ and $d$~\cite{CuKoiSma99,Len99,KaKoi05,KaKoi06}. These results are based on Gap Theorems showing that the desired factors of a polynomial $f=\sum_{j=1}^k c_j \bm X^{\bm\alpha_j}$ must divide both $\sum_{j=1}^\ell c_j \bm X^{\bm\alpha_j}$ and $\sum_{j=\ell+1}^k c_j \bm X^{\bm\alpha_j}$ for some index $\ell$. 
This allows to reduce the computation to the case of low-degree polynomials, for which one applies the classical algorithms. These Gap Theorems are based on number-theoretic results.

We recently proposed a new approach for this problem and gave a new algorithm for the computation of the multilinear factors of multivariate lacunary polynomials~\cite{ChaGreKoiPoStr13,ChaGreKoiPoStr14}. The algorithm we obtained is simpler and faster than the previous ones. Moreover, since it is not based on number-theoretic results, it can be used for a larger range of fields, for instance for absolute  or approximate factorization, or for finite fields of large characteristic. 
In this paper, we propose a generalization of this algorithm to the case of factors of degree at most $d$. We briefly explain the new approach in the simplest case of linear factors of bivariate polynomials. 

Let $f=\sum_{j=1}^k c_j X^{\alpha_j}Y^{\beta_j}\in\KK[X,Y]$ for some field $\KK$ of characteristic $0$, with $\alpha_j\le\alpha_{j+1}$ for all $j<k$. A linear polynomial $(Y-uX-v)$ divides $f$ if and only if $f(X,uX+v)=0$. We proved that for $uv\neq0$, if $f(X,uX+v)$ is nonzero then its valuation, that is the largest power of $X$ dividing it, is bounded by $\alpha_1+\binom{k}{2}$. From this, we deduced a Gap Theorem: Suppose that there exists an index $\ell<k$ such that $\alpha_{\ell+1}>\alpha_1+\binom{\ell}{2}$ and let $f_1=\sum_{j=1}^\ell c_j X^{\alpha_j}Y^{\beta_j}$ and $f_2=\sum_{j=\ell+1}^k c_j X^{\alpha_j}Y^{\beta_j}$. Then for all $uv\neq0$, $f(X,uX+v)=0$ if and only if $f_1(X,uX+v)=f_2(X,uX+v)=0$. 
In other words, $(Y-uX-v)$ divides $f$ if and only if it divides both $f_1$ and $f_2$. From this Gap Theorem, an algorithm for computing linear factors $(Y-uX-v)$ with $uv\neq0$ follows quite easily: Apply the Gap Theorem recursively to express $f$ as a sum of low-degree polynomials, and compute their common linear factors using any classical factorization algorithm. The computation of the remaining possible linear factors such as $(Y-uX)$ or $(X-v)$ reduces to univariate lacunary factorization. 

To use the same strategy with degree-$d$ factors, we need some new ingredients. First, we view a degree-$d$ bivariate irreducible polynomial $g\in\KK[X,Y]$ as a polynomial in $Y$ whose coefficients are polynomials in $X$. The roots of $g$ can be expressed in an algebraic closure of $\KK[X]$ using the notion of \emph{Puiseux series}. If $\phi$ is such a root of $g$, then $g$ divides $f\in\KK[X,Y]$ if and only if $f(X,\phi)=0$. We give a bound on the valuation of such an expression where $f$ is a lacunary polynomial. This yields a new Gap Theorem. Yet, the bound and the Gap Theorem depend on the valuation of the root $\phi$ itself. This means that there are actually as many Gap Theorems as there are possible valuations of $\phi$. A second ingredient is the use of the \emph{Newton polygon} of $f$ to a priori compute these valuations. As in the case of linear factors, there are some special cases reducing to the univariate case, namely the \emph{weighted homogeneous} factors. 
The computation of these factors too makes use of the Newton polygon of $f$.

In what follows, we give two algorithms: We first show how to compute the weighted ho\-mo\-ge\-neous factors, given an oracle to compute the irreducible factors of degree at most $d$ of a univariate lacunary polynomial. The second algorithm reduces the computation of the other factors, called \emph{inhomogeneous}, to the factorization of some bivariate low-degree polynomials. 

Using both algorithms yields our first main result. 

\begin{theorem}\label{thm:bivar}
Let $\KK$ be any field of characteristic $0$. Given a lacunary polynomial $f\in\KK[X,Y]$ of degree $D$ with $k$ non\-zero terms and an integer $d$, the computation of the irreducible factors of degree at most $d$ of $f$, with multiplicity, reduces to
\begin{itemize}
\item the computation of the irreducible factors of degree at most $d$ of $k/2$ lacunary polynomials of $\KK[X]$ plus $d^{\O(1)}$ bit operations per factor in post-processing, and
\item the factorization of $\O(k^3)$ polynomials of $\KK[X,Y]$ of total degree sum at most $\O(d^4k^4)$,
\end{itemize}
plus at most $(k\log D+d)^{\O(1)}$ bit operations.
\end{theorem}

In the multivariate case, we cannot directly apply the same algorithm as in the bivariate case since the resulting algorithm would be exponential in the number of variables. Nevertheless, we can prove the following result.

\begin{theorem}\label{thm:multivar}
Let $\KK$ be any field of characteristic $0$. Given a lacunary polynomial $f\in\KK[X_1,\dotsc,X_n]$ of degree $D$ with $k$ nonzero terms and an integer $d$, the computation of the irreducible factors of degree at most $d$ of $f$, with multiplicity, reduces to
\begin{itemize}
\item the computation of the irreducible factors of degree at most $d$ of $(nk)^{\O(1)}$ lacunary polynomials of $\KK[X]$ plus $(nd)^{\O(1)}$ bit operations per factor in post-processing, and
\item the factorization of $k$ polynomials of $\KK[X_1,\dotsc,X_n]$ of total degree sum at most $(nk\log(D)+d)^{\O(1)}$,
\end{itemize}
plus at most $(nk\log D+d)^{\O(1)}$ bit operations.
\end{theorem}

In the case of number fields, Lenstra gave a polynomial-time algorithm to compute the factors of degree at most $d$ of a univariate lacunary polynomial~\cite{Len99}. For the low-degree factorization of a polynomial $g$, there exist deterministic algorithms that run in time $(\size(g)+\deg(g))^{\O(n)}$ and return the list of factors in lacunary representation~\cite{Kal85}, and randomized algorithms that run in time $(\size(g)+\deg(g))^{\O(1)}$ and return the factors as straight-line programs~\cite{Kal89} or blackboxes~\cite{KaTra90}. As a result, we obtain a new algorithm for the computation of factors of degree at most $d$ of multivariate lacunary polynomials over number fields, giving a new proof of the main result of~\cite{KaKoi06}.

\begin{corollary}
There exists an 
algorithm that, given as inputs an irreducible polynomial $\varphi\in\QQ[\xi]$ representing a number field $\KK=\QQ[\xi]/\langle\varphi\rangle$, a lacunary polynomial $f\in\KK[X_1,\dotsc,X_n]$ and an integer $d$, computes the lacunary representation of the irreducible factors of degree at most $d$ of $f$, with multiplicity, in deterministic time $(\size(f)+d)^{\O(n)}$.

If the factors are represented as straight-line programs or blackboxes, the algorithm is randomized and runs in time $(\size(f)+d)^{\O(1)}$.
\end{corollary}

Note that Theorems~\ref{thm:bivar} and \ref{thm:multivar} are valid with any field of characteristic $0$. For instance, 
as long as a polynomial-time algorithm is known for multivariate low-degree factorization, we obtain a polynomial-time algorithm to compute the inhomogeneous low-degree factors of multivariate lacunary polynomials. These fields include the algebraic closure $\closure\QQ$ of $\QQ$ (absolute factorization~\cite{CheGa05}), the fields of real or complex numbers (approximate factorization~\cite{KaMayYaZhi08}), or the fields of $p$-adic numbers~\cite{Chi94}. Note that for $\closure\QQ$ and $\CC$, one cannot expect to have a polynomial-time algorithm computing all low-degree factors of multivariate lacunary polynomials since as we shall see, the computation of weighted homogeneous factors is equivalent to univariate lacunary factorization. The number of irreducible linear factors of a univariate polynomial over an algebraically closed field equals its degree, thus there are too many weighted homogeneous factors to compute them in polynomial time. The case of polynomials with real (approximate) coefficients is open to the best of my knowledge. In this case, Descartes' rule of signs implies that the number of real roots, or linear factors, is bounded by $2k-1$ where $k$ is the number of terms. Therefore, it is an intriguing question whether these roots can be computed in polynomial time.

We conjecture that the reductions presented in the current paper are valid in large positive characteristic, as in~\cite{ChaGreKoiPoStr13,ChaGreKoiPoStr14}, 
using Hahn series rather than Puiseux series. 
Another intriguing question is the validity of the approach in fields of small positive characteristic.

\paragraph{Organization}
In \refsec{NewtonPuiseux}, we collect some known facts about Newton polygons and Puiseux series. \refsec{quasihom} is devoted to the computation of the weighted homogeneous factors and \refsec{nonhom} to the computation of inhomogeneous factors, both for bivariate polynomials. In \refsec{multivariate} we give a proof sketch for the case of multivariate polynomials.

\paragraph{Acknowledgments}
I am grateful to P.~Koiran, N.~Portier and Y.~Strozecki for the numerous discussions we had on this work. I also wish to thank A.~Bostan, P.~Lairez J.~Le Borgne, B.~Salvy and T.~Vaccon for their help with Puiseux series, and the anonymous reviewers for their remarks which improved the presentation of this paper.


\section{Newton polygons and Puiseux series} \label{sec:NewtonPuiseux} 

We recall a few facts about Newton polygons and Puiseux series. For more on this topic, we refer the reader to~\cite{Stu98,Abhyankar}.

Let $f=\sum_j c_j X^{\alpha_j} Y^{\beta_j}$. Its \emph{support} is the set $\Supp(f)=\{(\beta_j,\alpha_j):c_j\neq0\}$. The \emph{Newton polygon} of $f$, denoted by $\Newt(f)$, is the convex hull of its support. 
Note that the coordinates are swaped in these two definitions compared to usual conventions.
For two convex polygons $A$ and $B$, their Minkowski sum is the set $A+B=\{a+b:a\in A,b\in B\}$.  

\begin{theorem}[Ostrowski]
Let $f$, $g$, $h\in\KK[X,Y]$ such that $f=gh$. Then $\Newt(f)=\Newt(g)+\Newt(h)$. 
\end{theorem}

We note that Ostrowski's Theorem was already used to compute the factorization of a polynomial using a decomposition of its Newton polygon~\cite{AbSaGaoLau04}. Yet computing such a decomposition is $\NP$-hard~\cite{GaoLau01}. This implies that this method has an inherent polynomial dependence on the degree of the polynomial to factor unless $\P=\NP$.

By contrast, we aim to obtain a logarithmic dependence on the degree. To this end we shall use the theorem to 
determine only some edges in the decomposition of the Newton polygon. 
We can see the Newton polygon of $f$ as a set of edges. By Ostrowski's Theorem, each edge of a factor of $f$ has to be parallel to an edge of $\Newt(f)$. Moreover, if we consider a degree-$d$ factor\footnote{From now on, the expression ``\emph{degree-$d$ factors}'' denotes the \emph{irreducible factors of degree at most $d$} of a polynomial.}  $g$ of $f$, its Newton polygon is inside a square whose sides have length $d$. For an edge of endpoints $(i,j)$ and $(i',j')$, its \emph{slope} is defined as $(j'-j)/(i'-i)$. Thus the slopes of the edges of $\Newt(g)$ have the form $p/q$, $p\in\ZZ$, $q\in\NN$, with $|p|,q\le d$. By convention, we say that a vertical edge has slope $-1/0$. In particular, only edges of $\Newt(f)$ with a slope $p/q$ with $|p|,q\le d$ can be edges of the Newton polygon of a factor of $f$.

Let $g\in\KK[X,Y]$, viewed as a polynomial in $Y$ with coefficients in $\KK[X]$. We are interested in the roots of $g$ in an algebraic closure of $\KK(X)$. This algebraic closure can be described using the \emph{field of Puiseux series over the algebraic closure $\closure\KK$ of $\KK$}, denoted by $\puiseux X$. Its elements are formal sums $\phi=\sum_{t\ge t_0} f_t X^{t/d}$ where $f_t\in\closure\KK$, $f_{t_0}\neq0$, $t_0\in\ZZ$ and $d\in\NN$. All we need for our purpose is that $\puiseux X$ \emph{contains} an algebraic closure of $\KK(X)$. In other words, any root of $g\in\KK[X][Y]$ can be described by a Puiseux series.

We define the valuation of a polynomial $f\in\KK[X]$ by $\val(f)=\max\{v:X^v\text{ divides }f\}$. This valuation is easily extended to the field of Puiseux series: If $\phi=\sum_{t\ge t_0} f_t X^{t/d}$ with $f_{t_0/d}\neq0$, then $\val(\phi)=t_0/d$. For bivariate polynomials in $\KK[X,Y]$, we define similarly the valuations with respect to $X$ and with respect to $Y$, and denote them by $\val_X$ and $\val_Y$ respectively.

Since $\puiseux X$ contains an algebraic closure of $\KK(X)$, a bivariate polynomial $g\in\KK[X][Y]$ of degree $d$ in $Y$ has exactly $d$ roots (counted with multiplicity) in $\puiseux X$. That is, there exist $\phi_1,\dotsc,\phi_d\in\puiseux X$ and $g_0\in\KK[X]$ such that
\[g(X,Y)=g_0(X) \prod_{i=1}^d (Y-\phi_i(X))\text.\]
The set $\{\val(\phi_i):1\le i\le d\}$ can be described in terms of the Newton polygon of $g$.

\begin{theorem}[Newton-Puiseux] 
Let $g\in\KK[X][Y]$. It 
has a root of valuation $v$ in $\puiseux X$ if and only if there is an edge of slope $-v$ in the lower hull of $\Newt(g)$.
\end{theorem}

The lower hull of $\Newt(g)$ is the set of edges of $\Newt(g)$ which are below $\Newt(g)$, excluding vertical edges. We define in the same way the upper hull of $\Newt(g)$.

As a consequence of Ostrowski's Theorem and Newton-Puiseux Theorem, we get informations on the roots of the factors of a polynomial by inspecting its Newton polygon.

\begin{corollary}
Let $f,g\in\KK[X,Y]$, where $g$ is a degree-$d$ factor of $f$. Then $g$ has a root $\phi\in\puiseux X$ of valuation $v$ only if there is an edge in the lower hull of $\Newt(f)$ of slope $-v=-p/q$ where $q>0$ and $|p|,q\le d$.
\end{corollary}

Let us suppose that we are given a bivariate lacunary polynomial $f=\sum_{j=1}^k c_jX^{\alpha_j} Y^{\beta_j}$ as the list of its nonzero terms, represented by triples $(c_j,\alpha_j,\beta_j)$. The common first step of our algorithms is the computation of the Newton polygon of $f$. This can be done in time polynomial in $k$ and $\log(\deg(f))$ using for instance Graham's scan~\cite{ComputationalGeometry}. The output is the ordered list of vertices of the Newton polygon. 


\section{Weighted homogeneous factors} \label{sec:quasihom} 

The aim of this section is to reduce the computation of the degree-$d$ weighted homogeneous factors of a bivariate lacunary polynomial to univariate lacunary factorization. 
We first collect some useful facts on weighted homogeneous polynomials. A polynomial $g=\sum_j b_j X^{\gamma_j}Y^{\delta_j}$ is said \emph{$(p,q)$-homogeneous of order $\omega$} if there exist two relatively prime integers $p$ and $q$, $q\ge 0$, and $\omega\ge0$ such that $p\gamma_j+q\delta_j=\omega$ for all $j$. In terms of the Newton polygon, this means that $\Newt(g)$ is contained in a line of slope $-q/p$. 
Note that there are two degenerate cases: $g$ is $(1,0)$-homogeneous if $\gamma_j$ is constant, thus if it can be written $X^\gamma h$ where $h\in\KK[Y]$, and similarly it is $(0,1)$-homogeneous if it can be written $Y^\beta h$ with $h\in\KK[X]$. 
Any polynomial $g$ can be written $g=g_1+\dotsb+g_s$ where the $g_t$'s are the $(p,q)$-homogeneous components of $g$, of pairwise distinct orders. 

The product of two $(p,q)$-homogeneous polynomials of order $\omega_1$ and $\omega_2$ respectively is $(p,q)$-homogeneous of order $\omega_1+\omega_2$. Conversely, any factor of a $(p,q)$-homogeneous polynomial is itself $(p,q)$-homogeneous. 

We shall also need a notion of $(p,q)$-homogenization of a univariate polynomial: If $p,q>0$, the $(p,q)$-homogenization of $h\in\KK[X]$ is $h_{p,q}=Y^{p\deg(h)} h(X^q/Y^p)$. A monomial $X^\delta$ of $h$ becomes $X^{q\delta}Y^{p(\deg(h)-\delta)}$. For all $\delta$, $p(\deg(h)-\delta)\ge 0$ and $p\cdot q\delta+q\cdot p(\deg(h)-\delta)=pq\deg(h)$ is independent of $\delta$. Thus $h_{p,q}$ is a $(p,q)$-homogeneous polynomial. If $p<0$ and $q>0$, the $(p,q)$-homogenization is defined by $h_{p,q}(X,Y)=h(X^q Y^{-p})$. Since $p<0$, $h_{p,q}$ is a polynomial and one easily checks that it is $(p,q)$-homogeneous of order $0$. The $(0,1)$-homogenization of $h\in\KK[X]$ is $h$ itself. The $(1,0)$-homogenization is only defined for $h\in\KK[Y]$ and is the identity too. It is clear that for all $p$ and $q$, the $(p,q)$-homogenization of a product $h_1h_2$ equals the product of the $(p,q)$-homogenizations of $h_1$ and $h_2$.

We define the \emph{normalization} of a bivariate polynomial $g$, denoted by $g^\circ$, as $g^\circ(X,Y)=X^{-\val_X(g)}Y^{-\val_Y(g)} g(X,Y)$, so that $\val_X(g^\circ)=\val_Y(g^\circ)=0$. Note that the $(p,q)$-homogenization of $h\in\KK[X]$ is a normalized polynomial, and every irreducible polynomial is in particular normalized. If $g=\sum_j b_jX^{\gamma_j}Y^{\delta_j}$ is normalized and $(p,q)$-homogeneous ($p,q\neq0$) of order $\omega$, then $q\divides\gamma_j$ and $p\divides\delta_j$ for all $j$. Indeed, since $g$ is normalized, there exists $j_1$ such that $\gamma_{j_1}=0$. Hence $q\delta_{j_1}=\omega$ and $q\divides\omega$. Let us thus write $\omega=q\omega'$. Now for all $j$, $p\gamma_j=q\omega'-q\delta_j$, whence $\gamma_j$ is divisible by $q$ since $p$ and $q$ are relatively prime. In the same way, $p$ divides $\delta_j$.

We first show that for all $p$ and $q$, one can reduce the computation of the degree-$d$ $(p,q)$-homogeneous factors of $f$ to the computation of the degree-$(d/q)$ factors of some univariate lacunary polynomials. 

\begin{theorem}\label{thm:quasihom}
Let $f=\sum_{j=1}^k c_j X^{\alpha_j}Y^{\beta_j}\in\KK[X,Y]$  
and let $f_1$, \dots, $f_s$ be its $(p,q)$-homogeneous components for some $p$ and $q$, $q\neq0$. 
Then $\mult{g}{f}=\min_{1\le t\le s}(\mult{g}{f_t})$ for any $(p,q)$-homogeneous irreducible polynomial $g$. 

Moreover, if $f_t^\circ$ denotes the normalization of $f_t$ for all $t$, 
\[\mult{g}{f_t}=\mult{g(X^{1/q},1)}{f_t^\circ(X^{1/q},1)}\text.\]
\end{theorem}

\begin{proof} 
If $f=g^\mu h$, one can write $h=h_1+\dotsb+h_{s'}$ as a sum of $(p,q)$-homogeneous components. Then each $g^\mu h_t$ is $(p,q)$-homogeneous, and they have pairwise distinct orders. Hence $s'=s$ and, up to reordering, $g^\mu h_t=f_t$ for all $t$. Therefore, $\mult{g}{f}\ge\min_t(\mult{g}{f_t})$. The converse inequality is obvious.

For the second part, let us assume that $f$ itself is a $(p,q)$-homogeneous and normalized polynomial. 
As mentioned earlier, $q\divides\alpha_j$. Therefore $f_q(X)=f(X^{1/q},1)$ and $g_q(X)=g(X^{1/q},1)$ are polynomials. 
Suppose that $f=g^\mu h$. Then $h$ is also $(p,q)$-homogeneous. Since $f$ is normalized, $h$ is normalized and the exponents of $X$ in $h$ are multiples of $q$. In other words, $f_q(X)=g_q(X)^\mu h(X^{1/q},1)$ is an equality of polynomials. Conversely, suppose that there exist $h_q$ such that $f_q(X)=g_q^\mu(X)h_q(X)$. To prove that $g^\mu$ divides $f$, it suffices to $(p,q)$-homogenize this equality. One can easily check that the $(p,q)$-homogenization of $f_q$ and $g_q$ are $f$ and $g$ respectively. Thus if we denote by $h$ the $(p,q)$-homogenization of $h_q$, $f=g^\mu h$. 
\end{proof} 

The case $q=0$ is similar. The only difference is for the second part of the theorem: The conclusion is $\mult{g}{f_t}=\mult{g}{f_t^\circ(1,Y)}$ since $g\in\KK[Y]$ and $g(1,Y)=g(X,Y)$.

We can now give an algorithm to compute the degree-$d$ weighted homogeneous factors of a bivariate lacunary polynomial, provided we dispose of an algorithm for the computation of the degree-$d$ factors of univariate polynomials. We assume that such an algorithm is given as an oracle.

\begin{algo}\label{algo:quasihom} ~\\
\textbf{Input:} A polynomial $f\in\KK[X,Y]$ given in lacunary representation and an integer~$d$.\\
\textbf{Output:} The list $L$ of the degree-$d$ weighted homogeneous factors of $f$, with their multiplicities.\\
\textbf{Oracle:} Given $f_0\in\KK[X]$ in lacunary representation and an integer $d$, computes the degree-$d$ factors of $f_0$.
\begin{enumerate}
\item Compute $\Newt(f)$ and initialize $L\gets\emptyset$.
\item For each pair of parallel edges in $\Newt(f)$, of slopes $-q/p$ with $|p|,q\le d$:\footnote{Vertical edges are said to have slope $-1/0$ by convention.}
    \begin{enumerate}
    \item Compute the $(p,q)$-homogeneous components $f_1$, \dots, $f_s$ 
        of $f$, and their normalizations $f_1^\circ$, \dots, $f_s^\circ$;
    \item For $t=1$ to $s$:
        \begin{enumerate} 
        \item Using the oracle, compute the degree-$(d/q)$ factors $(h_1,\mu_1)$,  \dots, $(h_{s_t},\mu_{s_t})$ of $f_t^\circ(X^{1/q},1)$, resp. $f_t^\circ(1,Y)$ if $q=0$;
        \item Let $L_t$ be the list of pairs $(g_u,\mu_u)$, $1\le u\le s_t$, such that $g_u$ is the $(p,q)$-homogenization of $h_u$ and $\deg(g_u)\le d$.
        \end{enumerate}
    \item $L\gets L\cup \bigcap_{t=1}^s L_t$. 
    \end{enumerate}
\item Return $L$.
\end{enumerate}
\end{algo}

In the algorithm, union and intersection are multisets operations: if $(g,\mu_1)\in L_1$ and $(g,\mu_2)\in L_2$, then $L_1\cup L_2$ contains $(g,\max(\mu_1,\mu_2))$ and $L_1\cap L_2$ contains $(g,\min(\mu_1,\mu_2))$.

\begin{proposition}
Algorithm~\ref{algo:quasihom} is correct. If the input polynomial has degree $D$ and $k$ terms, the algorithm uses at most $(k\log D+d)^{\O(1)}$ bit operations, plus $d^{\O(1)}$ per factor in post-processing. The sum of the sizes of all the univariate lacunary polynomials given to the oracle is at most $\frac{k}{2}\size(f)$.
\end{proposition}

\begin{proof} 
A $(p,q)$-homogeneous polynomial has a Newton polygon contained in a line of slope $-q/p$. By Ostrowski's Theorem, $f$ can have a $(p,q)$-homogeneous degree-$d$ factor only if its Newton polygon has two parallel edges of slopes $-q/p$ with $|p|,q\le d$. (There is a special case: $(0,1)$-homogeneous factors are factors depending only on the variable $X$ and correspond to vertical edges.) Therefore, the set of pairs $(p,q)$ is correctly computed. 

Now for each such pair  $(p,q)$, 
the algorithm computes the $(p,q)$-homogeneous factors of $f$ of degree $d$. The correctness of this part directly follows from Theorem~\ref{thm:quasihom}. 
It is enough for the oracle to compute degree-$(d/q)$ factors since for a degree-$d$ factor $g$ of $f_t$, $g(X^{1/q},1)$ has degree $d/q$. Note that the factors we compute may still be of degree larger than $d$, hence we discard the higher-degree factors. 

All the steps are easily seen to be polynomial-time computable since they consist in simple manipulations of lists of integer exponents, including the computation of the Newton polygon as noticed at the end of \refsec{NewtonPuiseux}. For each factor, the post-proceessing step is a computation on a list of exponents of size at most $d^{O(1)}$.

There are at most $k/2$ pairs of parallel edges in $\Newt(f)$. For each such pair, since $f_1$, \dots, $f_s$ have a lacunary representation, $\sum_t\size(f_t)=\size(f)$, whence the result.
\end{proof} 


\section{Inhomogeneous factors} \label{sec:nonhom} 

In this section, we study the factors of a bivariate lacunary polynomial whose Newton polygon is not contained in a line, that is which are not weighted homogeneous. For a bivariate lacunary polynomial $f$ and an irreducible polynomial $g$ having a root $\phi\in\puiseux X$, we first give a bound on the valuation of $f(X,\phi(X))$ in the first section. In the second section, we use this bound to give a Gap Theorem for inhomogeneous degree-$d$ factors of bivariate lacunary polynomials. We deduce an algorithm to reduce the computation of these factors to some bivariate low-degree factorizations.

\subsection{Bounds on the valuation} 

The aim of this section is to prove the following theorem. 

\begin{theorem} \label{thm:val}
Let $g\in\KK[X][Y]$ be an irreducible polynomial of total degree $d$ such that $\frac{\partial g}{\partial Y}\neq0$, 
and $\phi\in\puiseux X$ be a root of $g$ of valuation $v$.

Let $f=\sum_{j=1}^\ell c_j X^{\alpha_j}Y^{\beta_j}$ be a polynomial with exactly $\ell$ terms,
and suppose that the family $(X^{\alpha_j}\phi^{\beta_j})_{1\le j\le\ell}$ is linearly independent.

Then
\[\val\bigl(f(X,\phi(X))\bigr)\le \min_{1\le j\le\ell}(\alpha_j+v\beta_j) + 
    (2d(4d+1)-v)
    \binom{\ell}{2}\text.\]
\end{theorem}

The proof of this theorem is based on the Wronskian of a family of series. 

\begin{definition}
Let $f_1,\dotsc,f_\ell\in\puiseux X$. Their \emph{Wronskian} is the determinant of the \emph{Wronskian matrix}
\[\wr(f_1,\dotsc,f_\ell)=\det\begin{bmatrix}
    f_1             & f_2               & \dotsb & f_\ell   \\
    f_1'            & f_2'              & \dotsb & f_\ell'  \\
    \vdots          & \vdots            &        & \vdots   \\
    f_1^{(\ell-1)}  & f_2^{(\ell-1)}    & \dotsb & f_\ell^{(\ell-1)} 
\end{bmatrix}.\]
\end{definition}

The main property of the Wronskian is its relation to linear independence. 
The following result is classical (see for instance~\cite{BoDu10}).

\begin{proposition}\label{prop:wronskian}
The Wronskian of $f_1,\dotsc,f_\ell$ is nonzero if and only if the $f_j$'s are linearly independent over $\closure\KK$.
\end{proposition}

We first need an easy lemma, already proved in~\cite{ChaGreKoiPoStr13,ChaGreKoiPoStr14} in the context of polynomials. The exact same proof remains valid with Puiseux series.

\begin{lemma} \label{lemma:valLowerBound}
Let $f_1$, \dots, $f_\ell\in\puiseux X$ be Puiseux series in the variable $X$. Then
\[\val(\wr(f_1,\dotsc,f_\ell))\ge\sum_{j=1}^\ell \val(f_j) -\binom{\ell}{2}\text.\]
\end{lemma}

We aim to upper bound the valuation of the Wronskian of the family $(X^{\alpha_1}\phi^{\beta_1},\dotsc,X^{\alpha_\ell}\phi^{\beta_\ell})$. We need first the following lemma, borrowed from~\cite{KoiPoTa13}.

\begin{lemma}\label{lemma:wronskian}
Let $g$, $\phi$ and $f$ be as in Theorem~\ref{thm:val}, and let $g_Y=\frac{\partial g}{\partial Y}$.
Then 
\[\wr(X^{\alpha_1}\phi^{\beta_1},\dotsc,X^{\alpha_\ell}\phi^{\beta_\ell})=X^{A-\binom{\ell}{2}} \phi^{B-\binom{\ell}{2}} \frac{h_\ell(X,\phi)}{g_Y^{\ell(\ell-1)}(X,\phi)}\]
where $A=\sum_j\alpha_j$, $B=\sum_j\beta_j$ and $h_\ell$ is a polynomial of degree $(1+2d)\binom{\ell}{2}$ in each variable. 
\end{lemma}

It remains to obtain a valuation bound for a Puiseux series in terms of a vanishing polynomial.

\begin{lemma}\label{lemma:polySeries}
Let $g$ and $\phi$ be as in Theorem~\ref{thm:val}. Let $h(X,Y)$ be a polynomial of degree at most $\delta$ in each variable. Then $\left|\val(h(X,\phi))\right|\le 2d\delta$.
\end{lemma}

\begin{proof} 
Let us consider the resultant 
\[r(X,Y)=\res_Z(g(X,Z),Y-h(X,Z)).\]
Then $r(X,h(X,\phi))=0$ vanishes since $\phi$ is a common root of both polynomials in the resultant.

Let us now consider the degree of $r$ in $X$. The coefficients of $g(X,Z)$ viewed as a polynomial in $Z$ have degree at most $d$ 
in $X$ by definition. In the Sylvester matrix, $\delta$ rows are made of the coefficients of $g$ since $Y-h(X,Z)$ has degree $\delta$ in $Z$. In the same way, the Sylvester matrix contains $d$ 
rows with the coefficients of $Y-h(X,Z)$, each of which has degree at most $\delta$ in $X$. Altogether, each term in the resultant has degree at most $2d\delta$ in $X$. 

We have shown that $h(X,\phi)$ is a Puiseux series which cancels a polynomial $r$ of degree at most $2d\delta$ in $X$. By Newton-Puiseux Theorem, the absolute value of its valuation is at most $2d\delta$. 
\end{proof} 

\begin{proof}[of Theorem~\ref{thm:val}] 
Let $W$ be the Wronskian of the family $(X^{\alpha_1}\phi^{\beta_1}, \dotsc, X^{\alpha_\ell}\phi^{\beta_\ell})$ and $\psi=f(X,\phi)$. Without loss of generality, let us assume that $\min_j(\alpha_j+v\beta_j)$ is attained for $j=1$. 

Using column operations on the Wronskian matrix, one can replace the first column by $\psi$ and its derivatives. The determinant of the new matrix is the Wronskian $W_\psi$ of $\psi$, $X^{\alpha_2}\phi^{\beta_2}$, \dots, $X^{\alpha_\ell}\phi^{\beta_\ell}$. We have $W_\psi=a_1W$ and their valuations coincide. By Lemma~\ref{lemma:valLowerBound}, 
\[\val(W_\psi)\ge\val(\psi)+\sum_{j>1} (\alpha_j+v\beta_j)-\binom{\ell}{2}\text.\]

On the other hand, since the family $(X^{\alpha_j}\phi^{\beta_j})_j$ is linearly independent, there exists a nonzero $h_\ell$ such that
\[W=X^{A-\binom{\ell}{2}}\phi^{B-\binom{\ell}{2}}\frac{h_\ell(X,\phi)}{g_Y^{\ell(\ell-1)}(X,\phi)}\]
according to Lemma~\ref{lemma:wronskian}. Moreover $\val(h_\ell(X,\phi))\le 2d(2d+1)\binom{\ell}{2}$ 
and $\val(g_Y(X,\phi))\ge -2d^2$ 
by Lemma~\ref{lemma:polySeries}. Therefore, 
\[\val(W)\le A-\binom{\ell}{2}+v\left(B-\binom{\ell}{2}\right)+2d(4d+1)\binom{\ell}{2} 
\text.\]

Since $A=\sum_j\alpha_j$ and $B=\sum_j\beta_j$, 
\[\val(\psi)\le \alpha_1+v\beta_1 -v\binom{\ell}{2}+2d(4d+1)\binom{\ell}{2}\text.\] 
The conclusion follows, since $\alpha_1+v\beta_1=\min_j(\alpha_j+v\beta_j)$.
\end{proof} 


\subsection{Gap Theorem and algorithm} 

\begin{theorem}[Gap Theorem]
Let $v\in\QQ$, $d\in\NN^\star$, and $f=f_1+f_2$, where
\[f_1=\sum_{j=1}^\ell c_jX^{\alpha_j} Y^{\beta_j}\quad\text{and}\quad f_2=\sum_{j=\ell+1}^k c_jX^{\alpha_j} Y^{\beta_j}\]
satisfy $\alpha_j+v\beta_j\le\alpha_{j+1}+v\beta_{j+1}$ for $1\le j<k$. Assume that $\ell$ is the smallest index, if it exists, such that
\begin{equation*} 
\alpha_{\ell+1}+v\beta_{\ell+1}>(\alpha_1+v\beta_1)+(2d(4d+1)-v)\binom{\ell}{2}\text.
\end{equation*}

Then for every irreducible polynomial $g$ of degree at most $d$ 
such that $g$ has a root of valuation $v$ in $\puiseux X$, 
\[\mult{g}{f}=\min(\mult{g}{f_1},\mult{g}{f_2})\text.\]
\end{theorem}

\begin{proof} 
Let us view $g$ as a polynomial in $\KK[X][Y]$, and let $\phi\in\puiseux X$ be a root of $g$ of valuation $v$. Then $g$ divides $f$ (resp. $f_1$, resp. $f_2$) if and only if $f(X,\phi)=0$ (resp. $f_1(X,\phi)=0$, resp. $f_2(X,\phi)=0$). And if $g$ divides both $f_1$ and $f_2$, it divides $f$. Let us assume that $g$ does not divide $f_1$ and prove that in such a case, it does not divide $f$ either. Let $\Delta=2d(4d+1)-v$. 

Since $g$ does not divide $f_1$, $f_1(X,\phi)$ is nonzero.  Let us consider a basis $(X^{\alpha_{j_t}}\phi^{\beta_{j_t}})_{1\le t\le m}$ 
of the family $(X^{\alpha_j}\phi^{\beta_j})_{1\le j\le\ell}$ and rewrite $f_1(X,\phi)$ as  
\[f_1(X,\phi)=\sum_{t=1}^m b_t X^{\alpha_{j_t}}\phi^{\beta_{j_t}}\]
where $b_1$, \dots, $b_m$ are linear combinations of $c_1$, \dots, $c_\ell$. Without loss of generality, we assume that $b_t\neq0$ for all $t$. Using Theorem~\ref{thm:val}, the valuation of $f_1(X,\phi)$ is bounded by $\alpha_{j_1}+v\beta_{j_1}+\Delta\binom{m}{2}$. Furthermore, by minimality of $\ell$, $\alpha_{j_1}+v\beta_{j_1}\le \alpha_1+v\beta_1+\Delta\binom{j_1-1}{2}$. Thus
\[\val(f_1(X,\phi))\le \alpha_1+v\beta_1+\Delta\left(\binom{j_1-1}{2}+\binom{m}{2}\right)\text.\]
Since $j_1+m-1\le\ell$, we deduce that $\val(f_1(X,\phi))\le\alpha_1+v\beta_1+\Delta\binom{\ell}{2}$ by superadditivity of the function $\ell\mapsto\binom{\ell}{2}$. 

Now, $\val(f_2(X,\phi))\ge\alpha_{\ell+1}+v\beta_{\ell+1}>\val(f_1(X,\phi))$ by hypothesis. Hence $f(X,\phi)=f_1(X,\phi)+f_2(X,\phi)$ cannot vanish. That is, $g$ does not divide $f$.

To obtain the conclusion on the multiplicity of $g$ as a factor of $f$, it remains to apply the same proof to the successive derivatives of $f$, $f_1$ and $f_2$. The point is that these derivatives have the same form as $f$, $f_1$ and $f_2$ if no term vanishes. This can be ensured by multiplying $f$ by large powers of $X$ and $Y$, without changing its non-monomial factors.
\end{proof} 

In the Gap Theorem, we assumed that $\alpha_j+v\beta_j\le\alpha_{j+1}+v\beta_{j+1}$ for all $j$. That is, we put an order on the monomials which depends on the value of $v$. Since we aim to use this theorem with different values on $v$, we restate it without referring to the order: Let $\I=\{1,\dotsc,k\}$, and suppose that $\I$ can be partitioned into $\I_1\sqcup\I_2$ such that $\I_1=\{i\in\I: \alpha_i+v\beta_i\le \min_j(\alpha_j+v\beta_j)+\Delta\binom{\ell}{2}\}$ where $\Delta=2d(4d+1)-v$. Then any degree-$d$ polynomial $g$ which has a root of valuation $v$ in $\puiseux X$ satisfies $\mult{g}{f}=\min(\mult{g}{f_{|\I_1}},\mult{g}{f_{|\I_2}})$, where $f_{|\I_1}=\sum_{j\in\I_1}c_j X^{\alpha_j}Y^{\beta_j}$ and $f_{|\I_2}$ is defined similarly.

It is straightforward to extend the Gap Theorem to a partition of $\I$ into subsets $\I_1$, \dots, $\I_s$, using recursion: Let us rename $\I_2$ into $\J$. Suppose we have partitioned $\I$ as $(\bigsqcup_{u=1}^t \I_u)\sqcup\J$. We can partition $\J=\J_1\sqcup\J_2$ using the Gap Theorem with $f_{|\J}$. Then let $\I_{t+1}=\J_1$ and $\J=\J_2$. When the Gap Theorem stops working because there is no more gap, let $\I_s=\J$. For all $t$ and all $j_1,j_2\in\I_t$,
\begin{equation*} 
\bigl|(\alpha_{j_1}+v\beta_{j_1})-(\alpha_{j_2}+v\beta_{j_2})\bigr|\le \Delta\binom{|\I_t|-1}{2}\text.
\end{equation*}

For $1\le t\le s$, let $f_t=f_{|\I_t}$. The previous construction together with the Gap Theorem ensures that $\mult{g}{f}=\min_t(\mult{g}{f_t})$. 
Our goal is to refine the partition of $\I$ into smaller subsets such that the polynomials obtained from this partition after normalization have low degree.

We first prove an easy lemma useful to give bounds in the next theorem.

\begin{lemma}\label{lemma:slopes}
Let $v_1=p_1/q_1$ and $v_2=p_2/q_2$ two rational numbers such that $0<p_1,q_1,p_2,q_2\le d$ and $v_1>v_2$. 

Then $1/(v_1-v_2)\le d^2$ and $(v_1+v_2)/(v_1-v_2)\le 2d^2$.
\end{lemma}

\begin{proof}
We have
\[\frac{p_1}{q_1}-\frac{p_2}{q_2}=\frac{p_1q_2-p_2q_1}{q_1q_2}\]
and since $v_1>v_2$, the numerator is a nonzero integer and $v_1-v_2\ge 1/d^2$. Similarly,
\[\frac{v_1+v_2}{v_1-v_2}=\frac{p_1q_2+p_2q_1}{p_1q_2-p_2q_1}\le 2d^2\text.\qedhere
\]
\end{proof}

\begin{theorem}\label{thm:partition}
Let $f,g\in\KK[X,Y]$ such that $f$ has $k$ monomials and $g$ has a degree $d$ and is not weighted homogeneous. There exists a deterministic algorithm that computes in time polynomial in $k$ and $d$ a set of at most $k$ polynomials $f_1^\circ$, \dots, $f_s^\circ$, such that each $f_t^\circ$ has $\ell_t$ nonzero terms, with $\sum_t\ell_t=k$, and degree at most $\O(d^4\binom{\ell_t-1}{2})$, and such that 
\[\mult{g}{f}=\min_{1\le t\le s}(\mult{g}{f_t^\circ})\text.\]
\end{theorem}

\begin{proof} 
Since $g$ is not weighted homogeneous, its Newton polygon is not contained in a line. Therefore, it has at least two non-parallel edges $e_1$ and $e_2$. The idea is to apply the Gap Theorem twice: first to $f$ with $e_1$ to get a partition $\I_1\sqcup\dotsb\sqcup\I_{s'}$ of $\I=\{1,\dotsc,k\}$, and then to each $f_t=f_{|\I_t}$ with $e_2$ to refine the partition. We shall then prove that this refined partition defines low-degree polynomials.  

There are three cases to handle: either $\Newt(g)$ has two edges in its lower hull, or it has two edges in its upper hull, or it has an edge in the lower hull and at least one vertical edge.
To simplify notations, let us define $D=2d(4d+1)$, $\Delta_1=D-v_1$ and $\Delta_2=D-v_2$.

The first case is simple. Let $-v_1$ be the slope of $e_1$ and $-v_2$ the slope of $e_2$, so that $g$ has a root of valuation $v_1$ and another one of valuation $v_2$ in $\puiseux X$. We can apply the Gap Theorem to $f$ with $v=v_1$ to partition $\I=\I_1\sqcup\dotsb\sqcup\I_t$, and then apply it to each $f_t=f_{|\I_t}$ with $v=v_2$ to partition each $\I_t$ as $\I_{t,1}\sqcup\dotsb\sqcup \I_{t,s_t}$. Consider one subset $\I_{t,u}$ and the corresponding polynomial $f_{t,u}=f_{|\I_{t,u}}$. Let us assume without loss of generality that $\alpha_i+v_i\beta_i=\min_{j\in\I_{t,u}}(\alpha_j+v_i\beta_j)$ for $i=1,2$. Then for all $j\in\I_{t,u}$ and for $i=1,2$, $\alpha_j+v_i\beta_j\le \alpha_i+v_i\beta_i+\Delta_i\binom{\ell}{2}$. Let $\ell_{t,u}=|\I_{t,u}|$. Then for all $p,q\in\I_{t,u}$,
\begin{align*}
\alpha_p-\alpha_q
    &=(\alpha_p-\alpha_1)+(\alpha_1-\alpha_q)\\
    &\le v_1(\beta_1-\beta_p)+\Delta_1\binom{\ell_{t,u}-1}{2}+v_1(\beta_q-\beta_1)\\
    &\le v_1(\beta_q-\beta_p)+\Delta_1\binom{\ell_{t,u}-1}{2}\text.
\end{align*}
This inequality still holds if we replace $v_1$ by $v_2$ and if $p$ and $q$ are exchanged. In other words, 
\[\alpha_q-\alpha_p\le v_2(\beta_p-\beta_q)+\Delta_2\binom{\ell_{t,u}-1}{2}\text.\]
We can sum both equations and reorganize to obtain
\[(\beta_p-\beta_q)(v_1-v_2)\le(\Delta_1+\Delta_2)\binom{\ell_{t,u}-1}{2}\text.\]
Since $p$ and $q$ can once again be exchanged, we conclude that for all $p$ and $q$,
\[\left|\beta_p-\beta_q\right|\le \frac{\Delta_1+\Delta_2}{\left|v_1-v_2\right|}\binom{\ell_{t,u}-1}{2}\text.\]
Using very similar arguments, one easily shows that
\[\left|\alpha_p-\alpha_q\right|\le \frac{|v_1|\Delta_2+|v_2|\Delta_1}{\left|v_1-v_2\right|}\binom{\ell_{t,u}-1}{2}\text.\]
By Lemma~\ref{lemma:slopes}, $|\alpha_p-\alpha_q|,|\beta_p-\beta_q|\le\O(d^4\binom{\ell_{t,u}-1}{2})$. Therefore the polynomial $f_{t,u}^\circ$ obtained after normalization of $f_{t,u}$ has $\ell_{t,u}$ nonzero terms and degree at most $\O(d^4\binom{\ell_{t,u}-1}{2})$. The theorem follows, with $s=\sum_t s_t$.

The next two cases actually reduce to the first one. For the second case, one can consider the reciprocals $f^X$ of $f$ and $g^X$ of $g$ with respect to the variable $X$, defined by $f^X(X,Y)=X^{\deg_X(f)}f(1/X,Y)=\sum_{j=1}^k c_j X^{\gamma_j}Y^{\beta_j}$ where $\gamma_j=\deg_X(f)-\alpha_j$ for all $j$ and similarly for $g^X$. Then $\mult{g}{f}=\mult{g^X}{f^X}$ and we can apply the first case since the lower hull of $\Newt(g^X)$ has two edges. The bounds on the degrees of the polynomials we obtain are still valid for their reciprocals.

The third case corresponds to $e_2$ being vertical. We simply invert the variables and consider $\bar f(X,Y)=f(Y,X)$ and $\bar g(X,Y)=g(Y,X)$. Then $\Newt(\bar g)$ must have two edges either in its lower hull or in its upper hull. This means that either the first or the second of the previous cases can be applied to $\bar f$ to still obtain the same bounds.
\end{proof} 

\begin{algo}\label{algo:nonhom} ~\\
\textbf{Input:} A polynomial $f\in\KK[X,Y]$ given in lacunary representation and an integer~$d$.\\
\textbf{Output:} The list $L$ of the degree-$d$ inhomogeneous factors of $f$, with their multiplicities.\\
\textbf{Oracle:} Given a degree-$\O(d^4k^2)$ polynomial $g\in\KK[X,Y]$, computes the irreducible factorization of $g$.
\begin{enumerate}
\item Compute $\Newt(f)$ and initialize $L\gets\emptyset$;
\item For each pair of non-parallel edges in $\Newt(f)$: 
    \begin{enumerate}
    \item Compute $f_1^\circ$, \dots, $f_s^\circ$ according to Theorem~\ref{thm:partition};
    \item For $t=1$ to $s$: Compute the list $L_t$ of degree-$d$ factors of $f_t^\circ$ using the oracle;
    \item $L\gets L\cup \bigcap_{t=1}^s L_t$. 
    \end{enumerate}
\item Return $L$.
\end{enumerate}
\end{algo}

\begin{proposition}
Algorithm~\ref{algo:nonhom} is correct. If $f$ has degree $D$ and $k$ nonzero terms, the algorithms uses at most $(k\log D+d)^{\O(1)}$ bit operations, and the sum of the degrees of the bivariate polynomials given to the oracle is at most $\O(d^4k^4)$.
\end{proposition}

\begin{proof}
The correctness follows from Ostrowski's Theorem and Theorem~\ref{thm:partition}. 
Furthermore, for each pair of edges, the polynomials $f_1^\circ$, \dots, $f_s^\circ$ have degree at most $\O(d^4\binom{\ell_t-1}{2})$ for all $1\le t\le s$, with $\sum_t\ell_t=k$. By superadditivity of the function $\ell\mapsto\binom{\ell}{2}$, $\sum_t\deg(f_t^\circ)\le\O(d^4\binom{k-1}{2})$. Since there are at most $\binom{k}{2}$ pairs of distinct edges, the result follows.
\end{proof}



\section{Multivariate polynomials} \label{sec:multivariate} 

To extend our method to multivariate polynomials $f\in\KK[X_1,\dotsc,X_n]$, a first idea consists in considering the $n$-dimensional Newton polytope of $f$. Yet the computation of the Newton polytope is not polynomial in $n$. Actually, we will use the $n(n-1)$ possible $2$-dimensional Newton polygons. For, we extend our definition of $\Newt(f)$: If $i_1\neq i_2$, $\Newt_{i_1,i_2}(f)$ is the Newton polygon of $f$ viewed as an element of $R[X_{i_1},X_{i_2}]$ where $R$ is the polynomial ring in the $(n-2)$ other variables over $\KK$. 

As for the case of bivariate polynomials, there exists a special case. This special case corresponds to factors $g$ whose $n$-dimensional support is contained in a line (and thus is $1$-dimensional). As for weighted homogeneous factors in the bivariate case, the computation of these factors reduces to univariate lacunary polynomials. Let us call these polynomials \emph{unidimensional polynomials}. Note first that for such a factor $g$, $\Newt_{i_1,i_2}(g)$ is contained in a line for all $i_1$ and $i_2$. Consider the Newton polygons $\Newt_{1,i}$ for all $i>1$. If $f$ has a unidimensional factor $g$ depending on $X_1$, there exists a corresponding pair of parallel edges in each $\Newt_{1,i}(g)$, which are horizontal if $g$ does not depend on $i$. Actually, these pairs of edges correspond to a same pair of edges in the $n$-dimensional Newton polytope of $g$. The algorithm to compute unidimensional factors \emph{depending on $X_1$} is as follows: Consider all the parallel edges in $\Newt_{1,2}(f)$. For each such pair, pick one of the edges (say in the lower hull or on the left if it is vertical) and denote by $(a_1,a_2)$ and $(b_1,b_2)$ its endpoints. Then, each $\Newt_{1,i}(f)$ should have an edge of endpoints $(a_1,a_i)$ and $(b_1,b_i)$ for some $a_i$ and $b_i$, as well as an edge parallel to this one if $a_i$ and $b_i$ are not both zero (in which case we are considering a factor which does not depend on $X_i$). Thus for each pair of parallel edges of $\Newt_{1,2}(f)$, we check if the corresponding edges exist in $\Newt_{1,i}(f)$ for $i>2$. Now if we view $f$ as a polynomial in $X_1$ and $X_2$, it is weighted homogeneous and we can apply the algorithm for bivariate polynomials to eliminate the variable $X_2$. In the same way we eliminate all the variables $X_i$ for $i=2$ to $n$ and we get univariate lacunary polynomials. If we have an oracle computing their low-degree factors, we can reconstruct, as in the bivariate case, the corresponding unidimensional factors, variable by variable. This gives all the factors depending on $X_1$. We apply the same algorithm forgetting the variable $X_1$ and replacing its role by $X_2$ to compute the factors depending on $X_2$ and not on $X_1$. We continue with all variables to get all the unidimensional factors. The running time of this algorithm is polynomial in $n$, $k$ and $\log(D)$ where $k$ is the number of nonzero terms in $f$ and $D$ its degree.

Let us now consider a \emph{multidimensional} factor $g$, that is a factor whose support is not contained in a line. Then for every variable $X_{i_1}$, there exists at least one variable $X_{i_2}$ such that $\Newt_{i_1,i_2}(g)$ is not contained in a line, but in one case: if $g$ does not depend on $X_{i_1}$. The idea of the algorithm is the following: For all variables $X_i$, $i>1$, consider the Newton polygons $\Newt_{1,i}(f)$. For each $i$, partition the set $\I=\{1,\dotsc,k\}$ into $\I_1\sqcup\dotsb\sqcup\I_s$ according to the pairs of non-parallel edges, as in the proof of Theorem~\ref{thm:partition}. Thus, we have $(n-1)$ partitions of $\I$. The idea is now to merge these partitions to build a single partition. For, suppose we have two partitions $\I=\bigsqcup_t\J^1_t$ and $\I=\bigsqcup_t\J^2_t$ that we want to merge. We define a new partition $\I=\bigsqcup_t\I_t$ recursively. Let $\I_1=\{1\}$. Then, for every $j\in\I_1$, if $j\in\J^1_t$ and $j\in\J^2_{t'}$, we replace $\I_1$ by $\I_1\cup\J^1_t\cup\J^2_{t'}$. Once every index $j$ in $\I_1$ has been treated, we take the smallest index $j\notin\I_1$ and define $\I_2=\{j\}$. We apply the same algorithm to $\I_2$ and recursively build a partition of $\I$. 

If two distinct indices $j_1$ and $j_2$ belong to a same subset $\J^i_t$ ($i=1$ or $2$) of a partition, we have $|\alpha_{1,j_1}-\alpha_{1,j_2}|\le Cd^4|\J^i_t|^2$ for some constant $C$ (cf. Theorem~\ref{thm:partition}). Consider then two indices $j_1$ and $j_2$ in a same subset $\I_t$ of the new partition. They can be joined by a path of indices such that two consecutive indices in this path belong to a same $\J^1_t$ or a same $\J^2_t$. In other words, there exist indices $u_1=j_1$, $u_2$, \dots, $u_{2m}=j_2$ such that $u_1,u_2\in\J^1_{t_1}$, $u_3,u_4\in\J^1_{t_3}$, \dots, $u_{2m-1},u_{2m}\in\J^1_{t_{2m-1}}$ on the one hand, and $u_2,u_3\in\J^2_{t_2}$, \dots, $u_{2m-2},u_{2m-1}\in\J^2_{t_{2m-2}}$ on the other hand, for some $t_1$, \dots, $t_{2m-1}$. Then,
\begin{align*}
|j_2-j_1|   &\le\sum_{p=1}^{2m-1} |\alpha_{1,u_p}-\alpha_{1,u_{p+1}}|\\
            &\le Cd^4(|\J^1_{t_1}|^2+|\J^2_{t_2}|^2+\dotsb+|\J^1_{t_{2m-1}}|^2)\text.
\end{align*}
We can assume without loss of generality that the $\J^1_{t_p}$'s are pairwise distinct, as well as the $\J^2_{t_p}$'s. 
Since the sum of the sizes of the $\J^1_t$'s, respectively of the $\J^2_t$'s, is bounded by $k$, and since the function $k\mapsto k^2$ is superadditive, $|\alpha_{1,j_2}-\alpha_{1,j_1}|\le 2Cd^4k^2$. 
This means that we can merge all partitions built using the Newton polygons $\Newt_{1,i}(f)$ to get a new partition $\I=\I_1\sqcup\dotsb\sqcup\I_s$ such that for all $t$ and $j_1,j_2\in\I_t$, $|\alpha_{1,j_1}-\alpha_{1,j_2}|\le\O(nd^4k^2)$. 

This new partition has the property that if we define the normalized polynomials $f_t^\circ=f_{|\I_t}^\circ$ for all $t$, then $\mult{g}{f}=\min_t(\mult{g}{f^\circ_t})$ for all degree-$d$ multidimensional polynomials depending on $X_1$. To include factors which do not depend on $X_1$, we simply have to ensure that two indices $j_1$ and $j_2$ such that $\alpha_{1,j_1}=\alpha_{1,j_2}$ belong to the same subset. This can be done by merging the partition $\I_1\sqcup\dotsb\sqcup\I_s$ with the partition induced by the equalities on $\alpha_{1,j}$. The bound on $|\alpha_{1,j_1}-\alpha_{1,j_2}|$ remains valid.

Now, we can replace $X_1$ by $X_2$ and refine the partition we have with the same algorithm, and so on with all variables. Let $\I=\I_1\sqcup\dotsb\sqcup\I_s$ be the final partition and let $f_t^\circ$ be the normalization of $f_{|\I_t}$ for all $t$. The degree of $f_t^\circ$ is at most $\O(nd^4k^2)$ in each variable, and for any irreducible multidimensional polynomial $g$ of degree at most $d$, $\mult{g}{f}=\min_t(\mult{g}{f^\circ_t})$. It only remains to factorize these low-degree polynomials.



\end{document}